\documentclass{amsart}
\usepackage{amsmath}
\usepackage{amssymb}
\usepackage{bm}

\newtheorem{lemma}{Lemma}
\newtheorem{theorem}{Theorem}
\newtheorem{definition}{Definition}

\def\cD{\mathcal{D}}

\newcommand{\nz}{\mathbb{N}} 
\newcommand{\cz}{\mathbb{C}} 

\def\gH{\mathfrak{H}}
\def\gQ{\mathfrak{Q}}

\def\gh{\mathfrak{h}}

\def\gF{\mathfrak{F}}
\def\gD{\mathfrak{D}}

\def\rd{\mathrm{d}}
\def\ri{\mathrm{i}}
\def\re{\mathrm{e}}

\def\p{{\mathbf{\hat p}}} 

\DeclareMathOperator{\supp}{supp}
\DeclareMathAlphabet{\mathsfsl}{OT1}{cmss}{m}{sl}
\newcommand{\tensor}[1]{\mathsfsl{#1}}

\title[DIRAC ON METRIC TREES]{THE DIRAC OPERATOR ON REGULAR METRIC TREES}

\author[X. LIU]{Xiao Liu \\ }

\begin{document}

\begin{abstract}
A metric tree is a tree whose edges are viewed as line segments of positive length. The Dirac operator on such tree is the operator which operates on each edge, complemented by the matching conditions at the vertices which were given by Bolte and Harrison \cite{BolteHarrison2003}. The spectrum of Dirac operator can be quite different, reflecting geometry of the tree. \\
We discuss a special case of trees, namely the so-called regular trees. They possess a rich group of symmetries. This allows one to construct an orthogonal decomposition of the space $L^2(\Gamma)$ which reduces the Dirac. Based upon this decomposition, a detailed spectral analysis of Dirac operator on the regular metric trees is possible.
\end{abstract}

\maketitle

\section{INTRODUCTION}
\label{sec1}
Over the last decade, there were several papers about operator analysis on graph theory. However Naimark and Solomyak \cite{NaimarkSolomyak2001} took a first featured step in this direction on the Laplace operator.

\quad

\section{REGULAR ROOTED METRIC TREES}
\label{sec2}

\subsection{Geometry of a tree\label{ss2.1}}
\ \\

Let $\Gamma$ be a rooted tree. Assume that $o$ is the root, $V=V(\Gamma)$ is the set of vertices and $E=E(\Gamma)$ is the set of edges of $\Gamma$. Suppose that $\#V=\#E=\infty$. Each edge $e$ of a metric tree is regarded as a non-degenerate line segment of length $|e|$. The distance $d(x,y)$ between any two points $x, y \in \Gamma$, and thus the metric topology on $\Gamma$, is introduced in a natural way. Everywhere below, $|x|$ stands for $d(x,o)$. A subset $\Omega \subset \Gamma$ is compact if and only if it is closed and has non-empty intersections with only a finite number of edges.\\
For any two points $x,y \in \Gamma$ there exists a simple polygonal path in $\Gamma$ which starts at $x$ and terminates at $y$. This path is unique and we denote it by $$<x,y>:=\{z\in\Gamma:x\preceq z\preceq y\},$$ where we write $x\preceq y$ if $|y|=|x|+d(x,y)$.\\ We write $x \prec y$ if $x \in <o,y>$ and $x \ne y$.  \\
For any vertex $v$ its generation $gen(v)$ is defined as $$gen(v)=\#\{x\in V(\Gamma):x\prec v\}.$$
In particular, $v=o$ is the only vertex such that $gen(v)=0$. For any edge emanating from the vertex $v$ (which means that $e=<v,w>$ and $v\prec w$) we define its generation as $gen(e):=gen(v)$. \\
The branching number $b(v)$ of a vertex $v$ is defined as the number of edge emanating from $v$. We assume that $gen(v)<\infty$ for any $v$, $b(o)=1$ and $b(v)>1$ for $v\ne o$. We denote by $e{^1_v}, ..., e{^{b(v)}_v}$ the edges emanating from $v\in V$, and by $e{^-_v}$ the only edge which terminates at a vertex $v\ne o$.

\begin{definition}
 \label{def:2.1}
 A tree $\Gamma$ is called to be regular if all the vertices of the same generation have equal branching numbers, and all the edges of the same generation are of the same length.
\end{definition}

In this paper we consider regular trees only. Clearly, any regular tree is fully determined by specifying two number sequences (generating sequences) $\{b_n\}=\{b_n(\Gamma)\}$ and $\{t_n\}=\{t_n(\Gamma)\}$, $n=0,1,...$ s. t. $$b(v)=b_{gen(v)},\  |v|=t_{gen(v)} \quad \forall v\in V(\Gamma).$$
According to our assumptions, one has $b_0=1$ and $b_n \geq 2$ for any $n>0$. It is clear that $t_0=0$ and the sequence $\{t_n\}$ is strictly increasing, and we write
\begin{equation}
  \label{eq:2.1-1}
  h(\Gamma)=\lim_{n\to \infty}t_n=\sup_{x\in \Gamma}|x|.
\end{equation}
It is natural to refer to $h(\Gamma)$ as the height of $\Gamma$.

Another useful characteristics of the regular tree is its branching function. The branching function $g_{\Gamma}(t)$ of $\Gamma$ is defined by $$g_{\Gamma}(t)=\#\{x\in \Gamma: |x|=t\}.$$
Obviously,
\begin{equation}
  \label{eq:2.1-2}
  g_{\Gamma}(0)=1; \qquad g_{\Gamma}(t)=b_0...b_n,\  t_n<t\leq t_{n+1},\  n=0,1,...
\end{equation}
We also introduce the reduced height of $\Gamma$
\begin{equation}
  \label{eq:2.1-3}
  L(\Gamma)=\int_0^{h(\Gamma)}\frac {\rd t}{g_{\Gamma}(t)}.
\end{equation}
Clearly $h(\Gamma)<\infty$ implies $L(\Gamma)<\infty$. For the trees of infinite height both $L(\Gamma)<\infty$ and $L(\Gamma)=\infty$ is possible.
\\
It follows from here that the natural measure $\rd x$ on $\Gamma$ is induced by the Lebesgue measure on the edges. The spaces $L^p(\Gamma)$ are understood as $L^p$-spaces with respect to this measure. We denote by $|\Omega|$ the measure of a (measurable) subset $\Omega \subset \Gamma$ and call the number $|\Gamma|$ the total length of $\Gamma$. It is clear that $d(x,y)=|<x,y>|$ for any pair of points $x,y \in \Gamma$ and that $|\Gamma|=\int_\Gamma g_{\Gamma}(t)\rd t$.

\subsection{Special subtrees of $\Gamma$\label{ss2.2}}
\ \\

Subtrees $T \subset \Gamma$ of the following two types play a special part in the further analysis. For any vertex $v$ and for any edge $e=<v,w>$, $v \prec w$ we set $$T_v=\{x\in\Gamma: x\succeq v\},\qquad T_e=e\cup T_w.$$
In particular, $T_o=\Gamma$ and $$T_v=\cup_{1\leq j\leq b(v)} T_{e{^j_v}}, \quad \forall v \in V(\Gamma).$$
Let $T \subset \Gamma$ be a subtree. The branching function $g_T(t)$ of $T$ is defined by $$g_{T}(t)=\#\{x\in T: |x|=t\},\  0\leq t<h(T).$$
Due to the regularity of $\Gamma$, all the subtrees $T_e$, $gen(e)=k$ can be identified with a single tree $\Gamma_k$ whose generating sequences are $$b_0(\Gamma_k)=1, \quad b_n(\Gamma_k)=b_{k+n}(\Gamma), \ n\in \nz;$$
$$t_0(\Gamma_k)=0, \quad t_n(\Gamma_k)=t_{k+n}(\Gamma)-t_k(\Gamma), \ n\in \nz.$$

An important introduction is followed from here that the branching function $g_{\Gamma_k}$ is given by
\begin{equation}
  \label{eq:2.4}
g_{\Gamma_k}(t)=\frac{g_{\Gamma}(t_k+t)} {b_0...b_k}=\frac{g_{\Gamma}(t_k+t)} {g_{\Gamma}(t_k+)}, \quad k=0,1,...
\end{equation}
Note also that any subtree $T_v$, $gen(v)=k$, can be identified with the union of $b_k$ copies of the tree $\Gamma_k$ emanating from the common root $v$.

\quad

\section{THE DIRAC OPERATOR ON A REGULAR TREE}
\label{sec3}

The notion of Dirac operator on any metric graph, in particular on a tree, is well known. Still, for the sake of completeness we present here the variational definition of the Dirac operator with matching conditions \cite{BolteHarrison2003} on a tree.

\subsection{The Dirac operator\label{ss3.1}}
\ \\

First we say that a scalar-valued function $f$ on $\Gamma$ belongs to the Sobolev space $H^1=\bigoplus_{j=1}^{J}H^1((0,L_j))$ if $f$ is continuous, $f=\begin{pmatrix}
f_1 \\
f_2
\end{pmatrix}\upharpoonright e\in H^1(e)\otimes\cz^2$ for each edge $e$, and $$\|f\|^2_{H^1}:=\sum_{j=1}^J \left[\int_0^{L_j} (|f'(x)|^2+|f(x)|^2)\rd x \right]<\infty.$$
Let $H^1_c$ stand for the set of all functions from $H^1$ having compact support.

Next we define
\begin{equation}
  \label{eq:3.1-1}
   D_c:= \alpha\cdot \frac c\ri\nabla + c^2\beta
\end{equation}
 is the Dirac operator. Note that we are using atomic units in this
research, i.e., $m_{el}=\hbar=el=1$. As usual, the two matrices
$\alpha$ and $\beta$ are the Dirac matrices for two-component spinors
in standard representation, explicitly

$$
  \alpha=\begin{pmatrix}
            0 & -\ri \\
            \ri & 0
            \end{pmatrix}, \qquad
  \beta=\begin{pmatrix}
            1 & 0 \\
            0 & -1
            \end{pmatrix},$$

Let $\gD:= \bigoplus_{j=1}^{J}(H^1((0,L_j))\otimes\cz^2)$,
where $j$ turns over all bonds. \\
At the root we have the boundary condition $f (o)=0$, where $f \in \gD$; for the matching conditions at $v\ne 0$, we can construct matrices
\begin{equation}
  \label{eq:3.1-4}
  \mathbb{A}^{(k)}=
  \begin{pmatrix}
    1 & -1 & 0 & \cdots & 0 & 0 \\
    0 & 1 & -1 & \cdots & 0 & 0 \\
    0 & 0 & 1 & \cdots & 0 & 0  \\
    \vdots & \vdots & \vdots & \ddots & \vdots & \vdots \\
    0 & 0 & 0 & \cdots & 1 & -1 \\
    0 & 0 & 0 & \cdots & 0 & 0
  \end{pmatrix}_{v_k\times v_k}
\quad  \mbox{and} \quad
  \mathbb{B}^{(k)}=
  \begin{pmatrix}
    \multicolumn{4}{c}{\raisebox{0.7ex}[0pt]{\huge0}} \\
    1 & 1 & \cdots & 1
  \end{pmatrix}_{v_k\times v_k}.
\end{equation}
We have the condition
\begin{eqnarray}
  \label{eq:3.1-5}
  & \mathbb{A}^{(k)}f{^{(k)}_1}+\mathbb{B}^{(k)}f{^{(k)}_2}=0 & \\ \mbox{with} \qquad & rank(\mathbb{A}, \mathbb{B})=2v_k \qquad \mbox{and} & \qquad \mathbb{A}^{(k)}{\mathbb{B}^{(k)}}^T=\mathbb{B}^{(k)}{\mathbb{A}^{(k)}}^T  \nonumber
\end{eqnarray}
which satisfies \cite[Formula (4.15)]{BolteHarrison2003},
where
\begin{equation}
  \label{eq:3.1-2}
  f^{(k)}=\begin{pmatrix}
f{^{(k)}_1} \\
f{^{(k)}_2}
\end{pmatrix} \in  \bigoplus_{j=1}^{v_k}(H^1((0,L_j))\otimes\cz^{2v_k})
\end{equation}
of boundary values at $k$,
\begin{equation}
  \label{eq:3.1-7}
  f{^{(k)}_1}=\begin{pmatrix}
  f^1_1(0) \\
  f^2_1(0) \\
  \vdots \\
  f^{b_k}_1(0) \\
  f^{v_k}_1(L_{v_k})
  \end{pmatrix} \qquad
  f{^{(k)}_2}=\begin{pmatrix}
  -f^1_2(0) \\
  -f^2_2(0) \\
  \vdots \\
  -f^{b_k}_2(0) \\
  f^{v_k}_2(L_{v_k})
  \end{pmatrix}
\end{equation}
and $v_k=b_k+1$. \\

Without loss of generality, we can choose the space
\begin{equation}
  \label{eq:3.1-6}
  \gQ:= [\chi_{(0,\infty)}(D_c)]\bigoplus_{j=1}^{J}(H^1((0,L_j))\otimes\cz^2),
\end{equation}
where
\begin{equation}
  \label{eq:3.1-3}
   \chi_{(0,\infty)}(D_c):=
         \begin{cases}
            1,  & (f, D_c f)>0;\\
            0,  & (f, D_c f)\leq0
         \end{cases}
\end{equation} for any $f \in \gD$, s. t. $f^{(k)} \in [\chi_{(0,\infty)}(D_c)]\bigoplus_{j=1}^{v_k}(H^1((0,L_j))\otimes\cz^{2v_k})$.

\subsection{Reduction of the Dirac\label{ss3.2}}
\ \\

Our further analysis is  based upon an orthogonal decomposition of the space $L^2(\Gamma)$ which,
for the case of regular trees, reduces the Dirac. Let us describe this decomposition. \\
Given a subtree $T \subset \Gamma$, we say that a function $u \in L^2(\Gamma)$ belongs to the set
(a closed subspace) $\gF_T$ if $$u(x)=0 \  \mbox{for \  ~$x\notin T$~}; \qquad u(x)=u(y) \qquad \mbox{if\  ~$x,y \in T$~ \  and \  ~$|x|=|y|$~}.$$
In particular, $\gF_\Gamma$ consists of all symmetric (i.e. depending only on $|x|$) functions from $L^2(\Gamma)$.\\
We need the subspaces $\gF_T$ associated with the subtrees $T_e$ and $T_v$, introduced in Subsection \ref{ss2.2}. To simplify our notations, we shall write
$\gF_e$, $\gF_v$ instead of $\gF_{T_e}$, $\gF_{T_v}$. It is clear that for each vertex $v \ne o$ the subspaces $\gF_{e{^j_v}}, j=1,...,b(v)$ are mutually
orthogonal and their orthogonal sum $\widetilde{\gF_v}$ contains $\gF_v$. Denote $$\gF{^\prime_v}=\widetilde{\gF_v}\ominus \gF_v.$$

\begin{theorem}
 \label{thm:3.2-1}
 Let $\Gamma$ be a regular metric tree and $b(o)=1$. Then the subspaces $\gF{^\prime_v}$, $v\in V(\Gamma)$ are mutually orthogonal and orthogonal to $\gF_\Gamma$. Moreover,
 \begin{equation}
  \label{eq:3.2-1}
  L^2(\Gamma)\otimes\cz^2=\left(\gF_\Gamma\otimes\cz^2\right) \oplus \sum_{v\in V(\Gamma)}\oplus \left(\gF{^\prime_v}\otimes\cz^2\right).
 \end{equation}

\end{theorem}

\begin{proof}
The below result 
\begin{equation}
  \label{eq:3.2-6}
  L^2(\Gamma)=\gF_\Gamma \oplus \sum_{v\in V(\Gamma)}\oplus \gF{^\prime_v}
\end{equation}
is an obviously consequence of \cite[Formula (2.15)]{NaimarkSolomyak2001}. \\
According to Equation (\ref{eq:3.2-6}) and the distributive law, we can get that
\begin{multline}
 \label{mu:3.2-2}
 L^2(\Gamma)\otimes\cz^2 = \left(\gF_\Gamma \oplus \sum_{v\in V(\Gamma)}\oplus \gF{^\prime_v}\right)\otimes\cz^2 \\ = \left(\gF_\Gamma\otimes\cz^2\right) \oplus \sum_{v\in V(\Gamma)}\oplus \left(\gF{^\prime_v}\otimes\cz^2\right)
\end{multline}
\end{proof}

Now we have to describe the parts of the Dirac in the subspaces $\gF_\Gamma$ and $\gF{^\prime_v}$. For this purpose, along with the operator $D_c$ on $\Gamma\otimes\cz^2$ let us consider the Dirac operator on each tree $\Gamma_k\otimes\cz^2$ defined in Subsection \ref{ss2.2}. Below we denote this operator by $D_k$. Consider also its part $\cD_k=D_k\upharpoonright\gF_{\Gamma_k}\otimes\cz^2$ which is a natural analog of the operator $\cD_0=D_c\upharpoonright\gF_\Gamma\otimes\cz^2$.\\

By \cite[Page 325]{NaimarkSolomyak2001}, we get that the subspace $\gF_{e{^{<b(v)>}_v}}$ coincides with $\gF_v$. \\
The following holds
\begin{lemma}
 \label{lm:3.2-1}
\begin{equation}
 \label{eq:3.2-2}
  \gF{^\prime_v}=\gF_{e{^{<1>}_v}}\oplus...\oplus\gF_{e{^{<b_k-1>}_v}},
\end{equation}
where $\gF_{e{^{<j>}_v}}, j=1, ..., b(v)$ is defined as \cite{NaimarkSolomyak2001}.
\end{lemma}

\begin{theorem}
 \label{thm:3.2-2}
 Let $v \in V(\Gamma)$ and $gen(v)=k>0$. Then the operator $D_c\upharpoonright\gF{^\prime_v}\otimes\cz^2$ is unitarily equivalent to the orthogonal sum of $(b_k-1)$ copies of the operator $\cD_k$.
\end{theorem}

\begin{proof}
According to Lemma \ref{lm:3.2-1} and the distributive law, we can get that
\begin{multline}
 \label{mu:3.2-1}
 \gF{^\prime_v}\otimes\cz^2 = (\gF_{e{^{<1>}_v}}\oplus...\oplus\gF_{e{^{<b_k-1>}_v}})\otimes\cz^2 \\ = (\gF_{e{^{<1>}_v}}\otimes\cz^2)\oplus...\oplus(\gF_{e{^{<b_k-1>}_v}}\otimes\cz^2)
\end{multline}
\end{proof}

Our next step is to understand the nature of each operator $\cD_k$. Below we introduce a family ${M_k}, k=0, 1, ...$ of operators acting in $L^2(t_k, h(\Gamma))\otimes\cz^2.$  \\
The operator $M_k$ acts on the domain $Dom(M_k)$ as
\begin{equation}
 \label{eq:3.2-3}
  (M_k\psi)(t)=\alpha \cdot \frac c\ri\psi'(t) + c^2\beta \psi(t),\  t\ne t_j,\  j\ge k.
\end{equation}
Denote $I_j=(t_{j-1}, t_j),\  j\in \nz$. $Dom(M_k)$ consists of all functions $\psi$ on $[t_j, h(\Gamma))$,
$\psi=\begin{pmatrix}
\psi_1 \\
\psi_2
\end{pmatrix}$
\quad s. t. $\psi\upharpoonright I_j \in H^1(I_j)\otimes\cz^2 \quad \forall j>k$,
\begin{equation}
 \label{eq:3.2-4}
 \sum_{j>k}\int_{I_j}(|\psi'|^2+|\psi|^2)\rd t< \infty;
\end{equation}
and the following boundary condition at $t_k$ and the matching conditions at the points $t_j, \ j>k$ are satisfied:
\begin{equation}
 \label{eq:3.2-5}
\begin{pmatrix}
\psi_1(t_j+) \\
\psi_1(t_j-)
\end{pmatrix}
=\tensor B \begin{pmatrix}
-\psi_2(t_j+) \\
\psi_2(t_j-)
\end{pmatrix},\
\tensor B=\begin{pmatrix}
 b{^{1/2}_j} & 0 \\
 0 & -b{^{-1/2}_j}
\end{pmatrix}.
\end{equation}
Without loss of generality, we can consider the positive parts
$$M{^+_k}:=\chi_{(0,\infty)}(M_k)M_k;\quad B_k:=\chi_{(0,\infty)}(\cD_k)\cD_k.$$
Next we will show that $B_k \sim M{^+_k}$ for each $k$; here and in the sequel the symbol "$\sim$" stands for the unitary equivalence. \\

\begin{lemma}
 \label{lm:3.2-2}
 For any $k=0, 1, ...$
\begin{equation}
 \label{eq:3.2-7}
 B_k  \sim M{^+_k}
\end{equation}
\end{lemma}

\begin{proof}
\begin{trivlist}
\item[(i)]
First we can construct a block diagonal matrix $U^{(k)}$ of elements of $SU(2)$,
\begin{equation}
  \label{eq:3.2-15}
  U^{(k)}=
  \begin{pmatrix}
  \re^{\ri\theta_1\begin{pmatrix}
            0 & -\ri \\
            \ri & 0
            \end{pmatrix}} &   &   \\
    & \ddots &   \\
    &        & \re^{\ri\theta_{v_k}\begin{pmatrix}
            0 & -\ri \\
            \ri & 0
            \end{pmatrix}}
  \end{pmatrix}_{v_k\times v_k}
\end{equation}
where $\theta_s, s=1,...,v_k$ are selected randomly at each vertex. \\
Let \begin{equation}
  \label{eq:3.2-16}
  \mathbb{\hat{A}}^{(k)}=U^{(k)}\mathbb{A}^{(k)}(U^{(k)})^{-1}
  \end{equation}
\begin{equation}
  \label{eq:3.2-17}
  \mathbb{\hat{B}}^{(k)}=U^{(k)}\mathbb{B}^{(k)}(U^{(k)})^{-1}.
\end{equation}
Then we can get that $\mathbb{\hat{A}}^{(k)}$ and $\mathbb{\hat{B}}^{(k)}$ are satisfied (\ref{eq:3.1-5}) according to \cite[Section 6]{BolteHarrison2003}. \\
Thus $\cD_k$ with the boundary condition (\ref{eq:3.1-5}) is permutation-invariant.\\
\item[(ii)] To prove $B_0  \sim M{^+_0}$ \\
To identify a function $f=\begin{pmatrix}
f_1 \\
f_2
\end{pmatrix} \in \gF_\Gamma\otimes\cz^2$ with the function $\phi_f=\begin{pmatrix}
\phi_{f_1} \\
\phi_{f_2}
\end{pmatrix}$ on $\bigoplus_{j=1}^{\infty}([t_{j-1}, t_j)\otimes\cz^2)$ satisfying $f(x)=\phi_f(t) \  \forall x\in\Gamma \ $ s. t. $|x|=t$. Then
\begin{equation}
 \label{eq:3.2-8}
 \sum_{\sigma=1}^2\int_\Gamma|f_\sigma(x)|^2\rd x=\sum_{\sigma=1}^2\int_0^{h(\Gamma)}|\phi_{f_\sigma}(t)|^2g_\Gamma(t)\rd t, \ \ \forall f=\begin{pmatrix}
f_1 \\
f_2
\end{pmatrix} \in \gF_\Gamma\otimes\cz^2.
\end{equation}
If $f(o)=0$, then $\phi_f$ is continuous on $[0, t_1)\otimes\cz^2$, satisfies $\phi_f(0)=0$. \\
Now we substitute $\begin{pmatrix}
\psi_1(t) \\
\psi_2(t)
\end{pmatrix}=\begin{pmatrix}
\sqrt{g_\Gamma(t)}\phi_1(t) \\
\sqrt{g_\Gamma(t)}\phi_2(t)
\end{pmatrix}$, for $t\ne t_1, t_2, ...$. \\
Clearly $\psi(0)=0$ (because $\phi=\begin{pmatrix}
\phi_1 \\
\phi_2
\end{pmatrix}\in H^1\big([0, h(\Gamma))\big)\otimes\cz^2$, $\phi(0)=0$); at any point $t_j,\  j\in \nz$, the function $\psi$ meets the matching condition
$$\begin{pmatrix}
\psi_1(t_j+) \\
\psi_1(t_j-)
\end{pmatrix}
=\tensor B \begin{pmatrix}
-\psi_2(t_j+) \\
\psi_2(t_j-)
\end{pmatrix}$$
which comes from the continuity of $\phi$ on $\bigoplus_{j=1}^{\infty}([t_{j-1}, t_j)\otimes\cz^2)$. \\
All these above show that $\|f\|_{L^2(\Gamma)}^2=\|\psi\|_{L^2(0, h(\Gamma))}^2$.\\
It follows that $B_0  \sim M{^+_0}$ with the boundary condition (\ref{eq:3.1-5}). \\
\item[(iii)] Without loss of generality, we can get $B_k  \sim M{^+_k}$ following (i) and (ii).
\end{trivlist}
\end{proof}

The outcome of the analysis in this section is the following result. Below $M{^+_k}^{[r]}$ stands for the orthogonal sum of $r$ copies of a self-adjoint operator $M{^+_k}$.
\\
\begin{theorem}
 \label{thm:3.2-3}
 Let $\Gamma$ be the regular tree with the generating sequences $\{b_n\}$ $(with\ \ b_0=1)$ and $\{t_n\}$. Then
 \begin{equation}
 \label{eq:3.2-12}
 B_c \sim M{^+_0} \oplus \sum_{k=1}^{\infty}\oplus M{^+_k}^{[b_0...b_{k-1}(b_k-1)]}
 \end{equation}
 where $B_c:=[\chi_{(0,\infty)}(D_c)]D_c$.
\end{theorem}

\begin{proof}
Following the content of \cite{NaimarkSolomyak2001} from Formula (2.7) to Formula (2.8), we can construct
\begin{equation}
  \label{eq:3.2-13}
  U=
  \begin{pmatrix}
  1 &   &   &   \\
    & \re^{s(2\pi\ri)/b_k} &   &   \\
    &   & \ddots &   \\
    &   &   & \re^{(b_k-1)s(2\pi\ri)/b_k}
  \end{pmatrix}_{b_k\times b_k}b_k^{-1/2}
\end{equation}
where $s=1,...,b_k$.\\
According to \cite[Section 2]{NaimarkSolomyak2001}, we get $$D_c\upharpoonright\gF{^\prime_v}\otimes\cz^2=U\cD_kU^{-1}.$$
Thus \begin{equation}
  \label{eq:3.2-14}
   D_c\upharpoonright\gF{^\prime_v}\otimes\cz^2 \sim \cD_k^{[b_0...b_{k-1}(b_k-1)]}.
  \end{equation}
  According to Lemmata \ref{lm:3.2-1} and \ref{lm:3.2-2}, we get Formula (\ref{eq:3.2-12}).
\end{proof}

\quad

\section{THE DIRAC OPERATOR ON REGULAR TREES OF INFINITE HEIGHT}
\label{sec4}
Our next result is quite elementary and its proof is standard. The result applies to Laplacian rather than to Dirac operator only, see \cite{Solomyak2002}. Still, below we formulate only the case of Dirac operator we are interested in in this paper.\\
According to the theorem of Friedrichs extension, we know that $B_c$ is self-adjoint on $\gD$. Following this, we can get the below theorem.

\begin{theorem}
 \label{thm:4.1}
 Let $\Gamma$ be a regular tree and $\sup_{e\in E(\Gamma)}|e|=\infty$. \\ Then $\sigma(B_c)=[c^2, \infty)$.
\end{theorem}

\begin{proof}
\begin{trivlist}
\item[(i)] To prove $\sigma(\Delta)=[0, \infty)$. 

It is enough to show that for any $r>0$ the point $\lambda=r^2$ belongs to the spectrum. For this purpose we fix a non-negative function $\eta\in C_0^\infty[1, 2]$ s. t. $\supp\{\eta\}\in [1, 2]$. Further, choose an edge $e_m\in E(\Gamma)$. In an appropriate coordinate system, $e_m$ can be identified with the interval $(0, l_m)$ where $l_m=|e_m|=2^m$. The function $f_m$ on $\Gamma$,
\begin{equation}
  \label{eq:4.1}
  f_m(t)=\left\{
   \begin{array}{ll}
     \frac{\re^{\ri rt}}{l_m^{1/2}}\eta(\frac t{l_m}) &\mathrm{on} \quad e_m, \\
              0                                       &\mathrm{otherwise},
   \end{array}
   \right.
\end{equation}
belong to $Dom(\Delta)$.\\ 
An elementary calculation shows that $\|(\Delta-\lambda)f_m\|\to0$, as $m\to\infty$.\\
According to the Weyl criterion, we can get $\lambda\in\sigma(\Delta)$.\\
Thus $\sigma(\Delta)=[0, \infty]$.\\
\item[(ii)] Following Appendix \ref{sa1} and spectral mapping theorem (See \cite{Lax2002}), we know that $$\sigma(B_c)=[c^2, \infty).$$
\end{trivlist}
\end{proof}

\quad

\appendix
\section{The Proof of $B_0= E_c(\p)$\label{sa1}}

Any $\psi\in \gH:= [\chi_{(0,\infty)}(D_c)](L^2(\Gamma)\otimes\cz^2)$ can be written as
\begin{equation}
  \label{eq:6}
  \displaystyle
  \psi :=
  \begin{pmatrix}
    {E_c(\p)+c^2\over N_c(\p)}u\\
    {c\p \alpha\over N_c(\p)}u
  \end{pmatrix}
\end{equation}
for some $u\in \gh:=L^2(\Gamma)$. Here,

$$\p:=-\ri\nabla,\ \ E_c(\p):= (c^2\p^2+c^4)^{1/2},\
N_c(\p):=[2E_c(\p)(E_c(\p)+c^2)]^{1/2}.$$  In fact, the map
\begin{equation}
  \label{eq:7}
  \begin{split}
    \Phi: \gh &\rightarrow \gH\\
  u&\mapsto
  \begin{pmatrix}
    \Phi_1u\\
    \Phi_2u
  \end{pmatrix}
:=
\begin{pmatrix}
    {E_c(\p)+c^2\over N_c(\p)}u\\
    {c\p \alpha\over N_c(\p)}u
  \end{pmatrix}
\end{split}
\end{equation}
embeds $\gh$ unitarily into $\gH$ and its restriction onto
$H^1(\Gamma)$ is also a unitary mapping to $\gH\cap
H^1(\Gamma)\otimes\cz^2$ (Evans et al. \cite{Evansetal1996}).

\begin{lemma}
  \label{a3}
  To prove that
  $$B_0= E_c(\p).$$
\end{lemma}

\begin{proof}
For any $\psi\in \gH$, we can get
  \begin{multline}
   \label{eq:an3}
   (\psi,B_0\psi)=(\psi,\cD_0\psi)\\
   =\begin{pmatrix}
   {E_c(\p)+c^2\over N_c(\p)}u & {c\p \alpha\over N_c(\p)}u
   \end{pmatrix}
   \left(\begin{array}{cc}
      c^2        & c\,\alpha \p\\
      c\, \alpha \p & -c^2
   \end{array}\right)
   \begin{pmatrix}
    {E_c(\p)+c^2\over N_c(\p)}u\\
    {c\p \alpha\over N_c(\p)}u
   \end{pmatrix}\\
   =\begin{pmatrix}
   \Phi_1u & \Phi_2u
   \end{pmatrix}
   \begin{pmatrix}
    c^2\,\Phi_1u+c\,\alpha \p\,\Phi_2u\\
    c\,\alpha \p\,\Phi_1u+(-c^2)\,\Phi_2u
   \end{pmatrix}\\
   =\left(u,\left(\Phi{_1^*}c^2\Phi_1+\Phi{_1^*}c\,\alpha \p\,\Phi_2+\Phi{_2^*}c\,\alpha \p\,\Phi_1-\Phi{^*_2}c^2\Phi_2\right)u\right)\\
   =(u,E_c(\p)u).
  \end{multline}
\end{proof}

\quad

\end{document}